\RequirePackage{amsmath}
\documentclass[envcountsect,envcountsame,runningheads]{llncs}
\usepackage{amsmath}
\usepackage[T1]{fontenc}
\usepackage{mathtools}
\usepackage{amssymb}
\usepackage{todonotes}
\usepackage{mathabx}
\usepackage{nccmath}
\usepackage{enumerate}

\newcommand{\zost}{\{0,1\}^*}

\renewcommand{\b}[1]{\b{#1}}
\renewcommand{\b}[1]{\boldsymbol{#1}}

\newcommand{\aj}{A^{(j)}}
\newcommand{\asj}{A_x^{(j)}}
\newcommand{\zp}{\mathbb{Z}^+}

\newcommand{\mybmod}{\!\bmod\!}

\newcommand{\zo}{\{0,1\}}

\newlength{\arrow}

\newcommand*{\N}{\mathbb{N}}

\renewcommand{\b}[1]{\b{#1}}
\renewcommand{\b}[1]{\boldsymbol{#1}}

\newcommand{\npref}{\sqsubsetneqq}

\spnewtheorem{constr}[theorem]{Construction}{\bfseries}{\itshape}
\spnewtheorem{observation}[theorem]{Observation}{\bfseries}{\itshape}
\spnewtheorem*{notation}{Notation}{\bfseries}{\itshape}

\DeclareMathOperator{\GAPS}{GAPS}
\DeclareMathOperator{\UE}{UE}

\usepackage[symbol]{footmisc}

\usepackage{mathtools}
\DeclarePairedDelimiter{\abs}{\lvert}{\rvert}

\usepackage{apptools}

\numberwithin{equation}{section}

\newcommand*{\mcup}{\medmath\bigcup}

\usepackage{etoolbox}

\newcommand{\repthanks}[1]{\textsuperscript{\ref{#1}}}
\makeatletter
\patchcmd{\maketitle}
  {\def\thanks}
  {\let\repthanks\repthanksunskip\def\thanks}
  {}{}
\patchcmd{\@maketitle}
  {\def\thanks}
  {\let\repthanks\@gobble\def\thanks}
  {}{}
\newcommand\repthanksunskip[1]{\unskip{}}
\renewcommand{\@fnsymbol}[1]{%
  \ensuremath{%
    \ifcase#1
    \or 1
    \or *
    \fi
  }%
}
\makeatother

\begin{document}

\title{An Ordinal Extension Characterization of the Regular Languages with Applications}
\titlerunning{Ordinal Extension Characterization of Regularity}

\author{
    Joe Clanin%
        \thanks{\email{jsc@iastate.edu}\quad\textsuperscript{2}\email{lutz@iastate.edu}\quad\textsuperscript{3}\email{marayman@iastate.edu}\quad\textsuperscript{4}This author deceased.}
        \fnmsep
        \repthanks{nsfnote} 
    \and 
    Jack H. Lutz%
        \textsuperscript{2}%
        \fnmsep
        \repthanks{nsfnote} 
    \and 
    Matthew Rayman%
        \textsuperscript{3}%
        \fnmsep
        \thanks{Research supported in part by National Science Foundation grants 1545028 and 1900716\protect\label{nsfnote}}
    \and 
    Giora Slutzki
        \textsuperscript{4}%
}

\authorrunning{Joe Clanin, Jack H. Lutz, Matthew Rayman and Giora Slutzki}

\institute{Department of Computer Science\\ Iowa State University, Ames, IA 50011, USA}

\maketitle

\begin{abstract}
Given a language $A$ and a positive integer $j$, a $j^{\text{th}}$ $A$-extension is a string $y$ for which there exists a string $x$ such that $y$ is the $j^{\text{th}}$ string satisfying $xy \in A$.

\hspace*{1.5em}In this paper we show that ordinal extensions ($j^{\text{th}}$ $A$-extensions for various $A$ and $j$) constitute a subtle, but remarkably useful concept.  We characterize the regular languages in terms of ordinal extensions, and we extract from this characterization two general methods for giving very easy proofs of nonregularity. 

\end{abstract}

\section{Introduction}

One of the most fundamental questions that can arise about a language (i.e., decision problem) $A \subseteq \Sigma^*$ is whether $A$ is \textit{regular}, that is, whether membership in $A$ can be decided by a finite-state automaton. In the 75 years since Kleene \cite{Klee51} introduced them, the regular languages have been characterized in numerous ways, several of which appear in standard textbook treatments, e.g., \cite{hopcroft2006introduction}, \cite{kozen1997automata}, \cite{sipser2013introduction}.

This paper presents a new, subtle, and surprisingly useful characterization of the regular languages, along with some sample applications.  In this introduction we give a preliminary overview of our characterization, deferring some technical details.  Given a language $A \subseteq \Sigma^*$ and a positive integer $j$, we define a $j^\text{th}$ {\it $A$-extension} to be a string $y \in \Sigma^*$ for which there exists a string $x \in \Sigma^*$ such that $xy$ is the $j^\text{th}$ string satisfying $xy \in A$.  We use the term {\it ordinal extensions} generically for $j^\text{th}$ $A$-extensions for various $j$ and $A$. (Note that the term ``ordinal'' here is used as in ordinary language (``first'', ``second'', ``third'', etc.) and not as in transfinite set theory.)

Our main theorem is the Ordinal Extension Characterization of Regularity, which says that a language is regular if and only if (1) it has bounded ordinal extensions, and (2) its ordinal extensions finitely determine its left quotients. These two conditions are defined precisely in Sections 2 and 6 of this paper, but not in this introduction.

The rest of this paper is organized as follows.  In Section 2 we define bounded ordinal extensions and prove that every regular language has bounded ordinal extensions.  In Section 3 we show how this fact enables one to give exceedingly simple proofs of the nonregularity of many languages that have been considered in textbooks and previous papers. These proofs show that the languages at issue have infinite ordinal extensions, which implies that they have unbounded ordinal extensions.  In Section 4 we show that infinite ordinal extensions is a strictly stronger condition than unbounded ordinal extensions.  In Section 5 we exhibit a nonregular language that has bounded ordinal extensions, thereby establishing that bounded ordinal extensions alone do not characterize regularity.  In Section 6 we define what it means for a language's ordinal extensions to finitely determine its left quotients, and we prove that every regular language has this property. We also show that this property alone does not characterize regularity. In Section 7 we show how this property can be used to prove more nonregularity results.  In Section 8 we complete our Ordinal Extension Characterization of Regularity by proving that every language that has bounded ordinal extensions and whose ordinal extensions finitely determine its left quotients is regular.  Section 9 contains our concluding remarks.
\section{Bounded Ordinal Extensions}
This section presents our first tool for proving nonregularity.
Let $\Sigma$ be a nonempty, finite alphabet with a fixed ordering of elements. Let strings in $\Sigma^*$ be ordered first by length and, within each length, lexicographically according to the fixed ordering of $\Sigma$. For each set $A \subseteq \Sigma^*$ and each positive integer $j \in \zp$ not exceeding $|A|$, the $j^{\text{th}}$ element of $A$ is thus unambiguously defined. We write $\lambda$ for the empty string, the $1^{\text{st}}$ element of $\Sigma^*$.

For each $A \subseteq \Sigma^*$ and $x \in \Sigma^*$, we use the standard notation $$A_x = \{y \in \Sigma^* \mid xy \in A\}$$ for the left quotient of $A$ by $x$, or the set of all \textit{$A$-extensions} of $x$. For each $A \subseteq \Sigma^*$, $x \in \Sigma^*$, and $j \in \zp$, define the set $\asj$ as follows. If $j \leq |A_x|$, then $\asj = \{y\}$, where $y$ is the $j^{\text{th}}$ element of $A_x$. If $j > |A_x|$, then $A_x^{(j)} = \emptyset$. Note that $|\asj| \leq 1$ in any case.

For each $A, B \subseteq \Sigma^*$ and $j \in \zp$, let $$A_B^{(j)} = \bigcup_{x \in B} \asj$$ be the set of all $j^\text{th}$ $A$-extensions of elements of $B$, and let $$\aj = A_{\Sigma^*}^{(j)}$$ be the set of all $j^\text{th}$ $A$-extensions.

We say that a language $A \subseteq \Sigma^*$ has \textit{bounded ordinal extensions} if there exists $m \in \zp$ such that, for all $j \in \zp$, $|\aj| \leq m$. Otherwise, we say that $A$ has \textit{unbounded ordinal extensions}.

We say that a language $A \subseteq \Sigma^*$ has \textit{infinite ordinal extensions} if there exists $j \in \zp$ such that $|\aj| = \infty$. Clearly, a language with infinite ordinal extensions must have unbounded ordinal extensions.

 Recall the following notation and terminology for an equivalence relation $\approx$ on $\Sigma^*$.

\begin{enumerate}
    \item The \textit{quotient} of $\Sigma^*$ by $\approx$ is the set $\Sigma^* / \approx$ of all $\approx$-equivalence classes.
    \item $\approx$ is \textit{right-invariant} if, for all $x,y,z \in \Sigma^*$, $x \approx y$ implies $xz \approx yz$.
    \item A language $A \subseteq \Sigma^*$ \textit{respects} $\approx$ if, for all $x,y \in \Sigma^*$, $x \approx y$ implies that $x \in A \Leftrightarrow y \in A$.
\end{enumerate}

    The Myhill-Nerode theorem implies that a language $A \subseteq \Sigma^*$ is regular if and only if $A$ respects some right-invariant equivalence relation $\approx$ on $\Sigma^*$ with $|\Sigma^*/\approx| < \infty$.
The following lemma is the crux of the proof of our main theorem.

\begin{lemma}\label{lem2.1} Let $\approx$ be a right-invariant equivalence relation on $\Sigma^*$, and let $A \subseteq \Sigma^*$. If $A$ respects $\approx$, then, for all $j \in \zp$, \begin{equation}\label{eq2.1} |\aj| \leq |\Sigma^*/\approx|. \end{equation}
\end{lemma}

\begin{proof}
    Assume the hypothesis, and let $j \in \zp$. For all $w,x \in \Sigma^*$, we have $$w \approx x \implies A_w = A_x \implies A_w^{(j)} = A_x^{(j)},$$ the first implication holding because $\approx$ is right-invariant and $A$ respects $\approx$. Hence, for each $\approx$-equivalence class $C \in \Sigma^*/\approx$ and each $x \in C$, $$A_C^{(j)} = A_x^{(j)}.$$ This implies that \begin{equation}\label{eq2.2} \abs{A_C^{(j)}} \leq 1
\end{equation} for each $C \in \Sigma^*/\approx$. We now have

\begin{equation*}
\begin{split}
\aj & = \bigcup_{x \in \Sigma^*} \asj \\
 & = \bigcup_{C \in \Sigma^*/\approx} \medmath{\bigcup_{x \in C} \asj}\\
 &= \bigcup_{C \in \Sigma^*/\approx} A_C^{(j)}
\end{split}
\end{equation*}
It follows by (\ref{eq2.2}) that (\ref{eq2.1}) holds.
\end{proof}

\begin{theorem}\label{property1}(Ordinal Extension Regularity Theorem) If a language $A \subseteq \Sigma^*$ is regular then it has bounded ordinal extensions.

\end{theorem}

\begin{proof}

Let $A\subseteq \Sigma^*$ be regular and consider the right-invariant equivalence relation $\approx$ used in the Myhill-Nerode theorem giving an $m\in \mathbb{Z}^+$ with
\begin{equation}\label{eq2.3}|\Sigma^*/\approx| = m.\end{equation}

Then for all $j\in \mathbb{Z}^+$ we have $|A^{(j)}|\leq m$ by Lemma \ref{lem2.1} and hence $A$ has bounded ordinal extensions.

\end{proof}

\section{Some Applications to Nonregularity}

By taking the contrapositive of Theorem \ref{property1} we obtain the following useful method for proving nonregularity of languages.

\begin{corollary}
If a language $A\subseteq \Sigma^*$ has unbounded ordinal extensions then it is not regular.
\end{corollary}

This section applies the above bounded ordinal extension nonregularity method to some well-known nonregular languages $A$. In each of these cases, the method is easily applied by showing that $|\aj| = \infty$ for some small value of $j$. Most of the following examples can be found in \cite{FrGa18}.

\begin{example}
The language $B = \{0^n1^n \mid n \in \mathbb{N}\}$ has no infinite regular subset.
\end{example}
\begin{proof}
Let $A$ be an infinite subset of $B$. Then $\{1^n \mid 0^n1^n \in A\} \subseteq A^{(1)}$ (because each such $1^n$ is the first $A$-extension of $0^n$), so $A$ is not regular.
\end{proof}

\begin{example}
The set $A = \{x \in \zost \mid x^R = x\}$ of all binary palindromes is not regular. (Here, $x^R$ is the \textit{reverse} of $x$, i.e., $x$ written backwards.)
\end{example}
\begin{proof}
Each $0^n$ is the first $A$-extension of $0^n1$, so $0^* \subseteq A^{(1)}.$
\end{proof}

\begin{example}
For each infinite set $I \subseteq \mathbb{N}$ and each $n \in I$, let $n^I$ be the least element of $I$ that is greater than $n$, and let $$\GAPS_I = \{n^I - n \mid n \in I\}.$$
If $\GAPS_I$ is infinite, then the language $B_I=\{0^n\mid n\in I\}$ has no infinite regular subset.
\end{example}
\begin{proof}
Assume that $\GAPS_I$ is infinite, and let $A$ be an infinite subset of $B_I$. Then $A=\{0^n \mid n \in J\}$ for some infinite $J \subseteq I$. Now $\{0^k \mid k \in \GAPS_J\} = \{ 0^{n^I -n} \mid n \in J\} \subseteq A^{(2)}$(because each such $0^{n^I - n}$ is the second $A$-extension of $0^n$, $\lambda$ being the first), and $\GAPS_J$ is infinite, so $|A^{(2)}| = \infty.$
\end{proof}

\begin{corollary}
The set $B = \{0^p \mid p \text{ is prime}\}$ has no infinite regular subset.
\end{corollary}

\begin{example}
The set $A = \{xx^{R}y\mid x,\ y \in \zo^+\}$ is not regular. (Here $\zo^+ = \zost \setminus \{\lambda\}$.)
\end{example}
\begin{proof}
$A^{(1)} \supseteq \{(10)^n0 \mid n \in \mathbb{N}\}$, because each $(10)^n0$ is the first $A$-extension of $(01)^n$.
\end{proof}

\begin{example}
The language $$A = \{0^m1^n \mid m,n \in \zp \text{ and } gcd(m,n) = 1\}$$ is not regular.
\end{example}
\begin{proof}
For prime $p$, $1^p$ is the second $A$-extension of $0^{(p-1)!}$ ($1$ being the first), so $\{1^p \mid p \text{ is prime } \} \subseteq A^{(2)}$.
\end{proof}

\begin{example}
    The language
    \[A=\{x=y+z\mid x,y,z\text{ are binary integers and }x\text{ is the sum of }y\text{ and }z\}\] over $\Sigma=\{0,1,+,=\}$
    is not regular \cite{sipser2013introduction}.
\end{example}
\begin{proof}
    For each $n$, the first extension of $10^n=0+$ is $10^n$, so $\{10^n | n\in\mathbb{N}\}\subseteq A^{(1)}$.
\end{proof}

\section{Unbounded versus Infinite Ordinal Extensions}

The main objective of this section is to establish the existence of languages that have unbounded, but not infinite, ordinal extensions.

\begin{constr}\label{constr5.1}
For each $j \in \mathbb{N}$, let $$t(j) = \sum_{k=1}^j k = \binom{j+1}{2}$$ be the $j^{\text{th}}$ triangular number. For each $n \in \zp$, let $t^{-1}(n)$ be the unique $j \in \mathbb{N}$ such that \begin{equation}\label{eq5.1} t(j-1) < n \leq t(j). \end{equation}

For each $n,m,k \in \zp$, define the strings $$x_n = 0^n1, \hspace*{5mm} y(m,k) = x_mx_k.$$ Define the languages $$B = \{x_ny(t^{-1}(n),n-t(t^{-1}(n)-1)) \mid n \in \zp\},$$ $$C =\{x_n y(m,1) \mid n,m \in \zp \text{ and } m < t^{-1}(n)\},$$ and let $$A = B \cup C.$$

\end{constr}

\begin{theorem}\label{thm5.2} The language $A$ of Construction \ref{constr5.1} has unbounded ordinal extensions, but not infinite ordinal extensions.
\end{theorem}

\begin{proof}
Throughout the proof we use the fact that $$X = \{x_n \mid n \in \zp\}$$ is a prefix set, i.e., that no element of $X$ is a prefix of another element of $X$.

For each $n \in \zp$, if we let $j = t^{-1}(n)$, then the $A$-extensions of $x_n$, in order, are \begin{equation}\label{eqA.x}
y(1,1),...,y(j-1,1),y(j,n-t(j-1)).
\end{equation}
(The first $j-1$ of these are $C$-extensions of $x_n$ and the last is a $B$-extension of $x_n$.)

To see that $A$ has unbounded ordinal extensions, let $k \in \zp$. For each $1 \leq i \leq j$, (\ref{eqA.x}) tells us that $y(j,i)$ is the $j^{\text{th}}$ $A$-extension of $x_n$ where $n = t(j-1)+i$. Since $y(j,1),...,y(j,j)$ are distinct, it follows that $|\aj| \geq j$. Since $j$ is arbitrary here, this confirms that $A$ has unbounded ordinal extensions.

To see that $A$ does not have infinite ordinal extensions, let $j \in \zp$. It suffices to show that \begin{equation}\label{eqA.x1}|\aj| < \infty. \end{equation}Recall that $X = \{x_n \mid n \in \zp\}$, and define the sets \begin{equation} Y = \{x \in \{0,1\}^* \mid \text{ there exist } n\in\zp \text{ and } w \in A \text{ such that } x_n \npref x \npref w\},  \notag \end{equation}
\begin{equation} Z = \{x \in \{0,1\}^* \mid A_x = \emptyset\}.  \notag \end{equation} It is clear that $$\{0\}^* \cup A \cup X \cup Y \cup Z = \{0,1\}^*,$$ whence \begin{equation}\label{eqA.x+2} \aj = \aj_{\{0\}^*} \bigcup \aj_A \cup \aj_X \cup \aj_Y \cup \aj_Z. \end{equation}

We examine the five sets in this union in turn.

\begin{enumerate}
    \item[(i)] The $j^{\text{th}}$ extension of \emph{any} string $x \in \{0\}^*$ is $0^{j-1}1y(1,1)$, so \begin{equation}\label{eqA.x+3} \aj_{\{0\}^*} = \{0^{j-1}1y(1,1)\}.
    \end{equation}
    \item[(ii)]\begin{equation}\label{eqA.X+4}
        \aj_A = \begin{cases}
        \{\lambda\} &\text{if } j=1\\
        \emptyset &\text{if } j>1.
        \end{cases}
    \end{equation}
    \item[(iii)] By (\ref{eqA.x}) we have \begin{equation}\label{eqA.x+5}\aj_X = \{y(j,i) \mid 1\leq i \leq j\}.
    \end{equation}
    \item[(iv)] It is clear that \begin{equation}\label{eqA.x+6}
     \aj_Y \subseteq D,
    \end{equation} where $D$ is the set of all $u \in \{0,1\}^*$ for which there exist $r,s \in \{1,...,j\}$ such that $\lambda \npref u \npref x_rx_s.$
    \item[(v)]
    Trivially, \begin{equation} \label{eqA.x+7}
    \aj_Z = \emptyset.
    \end{equation}
\end{enumerate}

Since the sets on the right-hand sides of (\ref{eqA.x+3}), (\ref{eqA.X+4}), (\ref{eqA.x+5}), (\ref{eqA.x+6}), and (\ref{eqA.x+7}) are all finite, (\ref{eqA.x+2}) tells us that (\ref{eqA.x1}) holds. 

\end{proof}

We conclude this section with a lemma that sometimes enables an easy proof that a given language does not have infinite ordinal extensions.

Given a language $A \subseteq \Sigma^*$, call a string $y \in \Sigma^*$ a \textit{universal $A$-extension}, and write $y \in \UE(A)$, if, for every $x \in \Sigma^*$, $xy \in A$.

\begin{lemma}\label{lem5.3} Let $A \subseteq \Sigma^*$ and $j \in \zp$. If $\abs{\UE(A)} \geq j$, then $|\aj| < \infty$.

\end{lemma}

\begin{proof}
Assume the hypothesis. Let $z$ be the $j^{\text{th}}$ element of $\UE(A)$, and let $S$ be the set of all strings in $\Sigma^*$ up to and including $z$. It suffices to show that $\aj \subseteq S$. For this, let $y \in \aj$. Then there exists $x \in \Sigma^*$ such that $y$ is the $j^{\text{th}}$ $A$-extension of $x$. Since $S$ contains at least $j$ $A$-extensions of $x$ (namely, the first $j$ universal $A$-extensions), it follows that $y \in S$.
\end{proof}

\begin{corollary}\label{cor5.4}
Let $A \subseteq \Sigma^*$. If $\UE(A)$ is infinite, then $A$ does not have infinite ordinal extensions.
\end{corollary}

The language $A$ of the following example was introduced by Kamae and Weiss \cite{KamWei75} in connection with the theory of normal numbers.

\begin{example}\label{ex5.5}
The language $$A = \{u110^n10^n \mid u \in \zost \text{ and } n \in \zp\}$$ does not have infinite ordinal extensions.
\end{example}
\begin{proof}
For all $n \in \zp$, $(11010)^n \in \UE(A)$, so this follows immediately from Corollary \ref{cor5.4}.
\end{proof}
We will prove the nonregularity of this language later using a tool from a later section.
\section{A Nonregular Language with Bounded Ordinal Extensions}
In this section we exhibit a nonregular language $A \subseteq \{0\}^*$ that has bounded ordinal extensions.

\begin{constr}\label{constr4.1}
For each set $I \subseteq \mathbb{N}$, let $$I' = \{3i \mid i \in \mathbb{N}\} \cup \{3i+1 \mid i \in I\} \cup \{3i+2 \mid i \in \mathbb{N} \setminus I\},$$ and let $$A = A[I] = \{0^n \mid n \in I'\}.$$
\end{constr}

Throughout this section, $I'$ and $A$ are defined from $I$ as in Construction \ref{constr4.1}. Our main techincal lemma concerning this construction is the following.

\begin{lemma}\label{lem4.2}
For all $I \subseteq \mathbb{N}$ and all $j \in \zp$, \begin{equation}\label{eq4.1}|\aj| \leq 3.\end{equation}
\end{lemma}

Before proving Lemma~\ref{lem4.2}, we use it to prove the main result of this section.

\begin{theorem}\label{thm4.3}
There is a nonregular language that has bounded ordinal extensions.
\end{theorem}
\begin{proof}
Let $I \subseteq \mathbb{N}$ be undecidable. Then $A = A[I]$ is clearly not regular. By Lemma \ref{lem4.2}, $A$ has bounded ordinal extensions.
\end{proof}

The rest of this section is devoted to proving Lemma \ref{lem4.2}. Our proof uses a two-parameter family of sets defined as follows.

\begin{constr}\label{constr4.4}
For each $m \in \{0,1,2\}$ and $j \in \zp$, define the set $B_m^j \subseteq \{0\}^*$ by the following three recursions on $j$.

\begin{enumerate}
\item[(0)]
\begin{flalign*}
  &B_0^1 = \{\lambda\} \tag{4.2}\\
  &B_0^j = \begin{cases}
  B_0^{j-1}\cdot\{0,0^2\} &\text{ if } j > 1 \text{ is even}\\
  B_0^{j-2}\cdot\{0^3\} &\text{ if } j > 1 \text{ is odd.}
  \end{cases}
  \tag{4.3}\\
\end{flalign*}
\item[(1)]\begin{flalign*}
  &B_1^1 = \{\lambda,\ 0\};\ \ B_1^2 = \{0^2\} \tag{4.4}\label{eq4.4}\\
  &B_1^j = \begin{cases}
  B_1^{j-2}\cdot\{0^3\} &\text{ if } j > 2 \text{ is even}\\
  B_1^{j-1}\cdot\{0,0^2\} &\text{ if } j > 2 \text{ is odd.}
  \end{cases}
  \tag{4.5}\label{eq4.5}\\
\end{flalign*}

\item[(2)]\begin{flalign*}
  &B_2^1 = \{\lambda,\ 0\};\ \ B_2^2 = \{0,0^2,0^3\} \tag{4.6}\label{eq4.6}\\
  &B_2^j = \begin{cases}
  B_2^{j-1}\cdot\{0^2\} &\text{ if } j > 2 \text{ is even}\\
  B_2^{j-1}\cdot\{0\} &\text{ if } j > 2 \text{ is odd.}
  \end{cases}
  \tag{4.7}\label{eq4.7}\\
\end{flalign*}
\end{enumerate}

Finally, for each $j \in \zp$, let

\begin{flalign*}
B^{(j)} = B_2^j. \tag{4.8} \label{eq4.8}
\end{flalign*}
\end{constr}

It is clear by inspection that

\begin{center}
\begin{align}
    |B^{(j)}| \leq 3 \tag{4.9}\label{eq4.9}
\end{align}\end{center}
holds for all $j \in \zp$. Hence, to prove Lemma \ref{lem4.2}, it suffices to prove the following.

\begin{lemma}\label{lem4.5}
For all $j \in \zp$, $A^{(j)} \subseteq B^{(j)}.$
\end{lemma}

The proof of Lemma \ref{lem4.5} appears in the Optional Technical Appendix. 
\section{Finite Determination of Left Quotients by Ordinal Extensions}

We say that the ordinal extensions of A {\it finitely determine the left quotients of} $A$ if there exists $k\geq 1$ such that for all $x,y\in\Sigma^*$, \[A_x^{(j)}=A_y^{(j)} \text{ for } 1\leq j\leq k \implies A_x=A_y.\]

\begin{theorem}\label{property2}(Finite Determination of Left Quotients Regularity Theorem)
    If a language $A\subseteq \Sigma^*$ is regular, then the ordinal extensions of $A$ finitely determine the left quotients of $A$.
\end{theorem}
\begin{proof}
    Suppose $A$ is regular. By the Myhill-Nerode theorem, $\{A_x|x\in\Sigma^*\}$ is finite. For any two distinct sets $A_x$ and $A_y$, there must be a first string that belongs to one and not the other. Take $k$ to be the maximum of the indices of all such distinguishing strings. This $k$ testifies to the finite determination of the left quotients of $A$.
\end{proof}

As is the case in applications of the ordinal extension nonregularity method, we typically make use of Theorem~\ref{property2} in the following contrapositive form.

\begin{corollary}
    If for each $k\geq 1$ there exist $x,y\in\Sigma^*$ such that $A_x^{(j)}=A_y^{(j)}$ for all $1\leq j\leq k$ but $A_x\neq A_y$, then $A$ is not regular.
\end{corollary}

\begin{remark}
    Finite determination of left quotients does not characterize regularity. To see this, note that $\{0^n1^n | n\in\N\}$ satisfies the finite determination condition with $k=2$.
\end{remark}
\section{More Applications to Nonregularity}

We now demonstrate a few applications of the finite determination of left quotients nonregularity method.

\begin{example}
    The language of Example~\ref{ex5.5} is not regular.
\end{example}
\begin{proof}
Fix $k\geq 1$. Note that \[A_{110^n1}=A\cup\{10^i10^i|i\geq 1\}\cup\{0^n\}.\] For sufficiently large distinct $n$ and $m$ we have $A_{110^n1}^{(j)}=A_{110^m1}^{(j)}$ for $j=1,\dots,k$, however $A_{110^n1}\neq A_{110^m1}$ as $0^n\in A_{110^n1}$ but $0^n\notin A_{110^m1}$.
\end{proof}

\begin{example}
    The language $A=\{0^n\mid n\text{ is not a perfect square}\}$ is not regular.
\end{example}
\begin{proof}
 Consider the strings $0^{i^2}$ and $0^{j^2}$ where $0<i<j<\infty$. Note that the first string that is an extension of exactly one of these two strings is $0^{2i+1}$. Moreover, the strings $0^n$ for $0<n <2i+1$ are all extensions for both of these strings coming before $0^{2i+1}$. Hence, $A$ is not regular.
\end{proof}
The following is an example of a language where it is not possible to directly apply the pumping lemma.
\begin{example}\label{neq01}
    The language $A=\{w\in \{0,1\}^* \mid |w|_0\neq|w|_1\}$ where $|w|_a$ is the number of occurrences of  $a$ in $w$ is not regular.
\end{example}
\begin{proof}
    Consider the strings $0^i$ and $0^j$ where $0<i<j<\infty$. Note that the first string that is an extension of exactly one of these two strings is $1^i$. Moreover, the strings $1^n$ for $0\leq n <i$ are all extensions for both of these coming before $1^i$. Thus, $A$ is not regular.
\end{proof}

\begin{corollary}
    The language A of Construction \ref{constr4.1} is not regular.
\end{corollary}
\begin{proof}
By Lemma \ref{lem4.2} we have that $|A^{(j)}|\leq 3$ for all $j\in \mathbb{Z}^+$. Assume for contradiction that the ordinal extensions of $A$ finitely determine its left quotients. Let $k$ be the parameter associated with this determination. Then by the pigeonhole principle there are two strings $x=0^n$ and $y=0^m$ with $A^{(i)}_x=A^{(i)}_y$ for all $i\leq k$ and ergo all $i\in \mathbb{Z}^+$. Without loss of generality assume there is a string $0^r\in A$ with $n<r<m$. Given the strings in $A$ up to $0^m$, one can compute the rest of $A$ by iteratively computing $A^{(i)}_x$ and adding $yA^{(i)}_x$ into $A$. This contradicts the undecidability of $A$ and the corresponding $I$.
\end{proof}

\section{Ordinal Extension Characterization of Regularity}

\begin{theorem}\label{maintheorem} (Ordinal Extension Characterization of Regularity) A language $A\subseteq\Sigma^*$ is regular if and only if the following two conditions hold.
    \begin{enumerate}[(i)]
        \item $A$ has bounded ordinal extensions
        \item The ordinal extensions of $A$ finitely determine the left quotients of $A$.
    \end{enumerate}
\end{theorem}

\begin{proof}

First suppose that $A$ is a regular language. Properties $(i)$ and $(ii)$ then follow by Theorem \ref{property1} and Theorem \ref{property2}

Conversely, suppose $A\subseteq \Sigma^*$ has properties $(i)$ and $(ii)$, with parameters $m$ and $k$ respectively. By the Myhill-Nerode theorem, it suffices to show that \[\left\{(A_x^{(1)},A_x^{(2)},\dots,A_{x}^{(k)})\hspace{0.5em}|\hspace{0.5em}x\in\Sigma^*\right\}\] is finite. As each $A_x^{(i)}$ is either empty or a singleton from a set of size at most $m$, we obtain 
\begin{equation}\label{ktuplebound}
\left|\left\{(A_x^{(1)},A_x^{(2)},\dots,A_{x}^{(k)})\hspace{0.5em}|\hspace{0.5em}x\in\Sigma^*\right\}\right|\leq (m+1)^k.
\end{equation}

\end{proof}
\section{Conclusion}

Two problems that we have left conspicuously open are to find natural (as opposed to contrived) examples of the languages shown to exist by Theorems \ref{thm5.2} and \ref{thm4.3}. More generally, does the Ordinal Extension Characterization of Regularity enable simpler or more instructive proofs of any of the many known general theorems concerning regular languages?
\section*{Acknowledgment}

The first three authors dedicate this paper to the memory of the fourth author, our treasured friend and teacher with whom the second author began this investigation.

\bibliographystyle{splncs04}
\bibliography{master.bib}
\appendix

\newpage
\section{Optional Technical Appendix}

This appendix is devoted to proving Lemma \ref{lem4.5}. For this we first develop useful properties of Construction \ref{constr4.4}.

\begin{lemma}\label{lemA.1}
\begin{flalign*}
     \hspace{3.2cm}&1.\text{For all } j \in \zp,\text{ } B_0^j \cdot \{0^3\} = B_0^{j+2}.& \tag{A.1} \label{A.1}\\
     \hspace{3.2cm}&2.\text{For all } j \in \zp,\text{ } B_1^j \cdot \{0^3\} = B_1^{j+2}.&\tag{A.2} \label{A.2}\\
     \hspace{3.2cm}&3.\text{For all } 2 \leq j \in \zp, \text{ } B_2^j \cdot \{0^3\} = B_2^{j+2}.& \tag{A.3} \label{A.3}
\end{flalign*}
\end{lemma}

\begin{proof}
1. If $j$ is odd then (\ref{lemA.1}) follows immediately from (4.3). If $j$ is even, then three applications of (4.3) tell us that $$B_0^{j+2} = B_0^{j+1}\cdot\{0,0^2\} = B_0^{j-1}\cdot\{0^4,0^5\} = B_0^j\cdot\{0^3\}. $$

2. This is exactly like $1$ above, but using (4.4) and (4.5) in place of (4.3).

3. If $j$ is odd, then two applications of (4.7) tell us that $$B_2^{j+2} = B_2^{j+1} \cdot \{0\} = B_2^j \cdot \{0^3\}.$$ \qed
\end{proof}

\begin{lemma}\label{lemA.2} For all $j \in \zp$, $$B_0^j \cup B_1^j \subseteq B_2^j.$$
\end{lemma}

\begin{proof}
We prove the inclusions $B_0^j \subseteq B_2^j$ and $B_1^j \subseteq B_2^j$ by separate inductions on $j$.

By (4.2) and (4.6), $B_0^1 \subseteq B_2^1$. By (4.3), (4.2), and (4.6), $B_0^2 \subseteq B_2^2$. Assume that $B_0^k \subseteq B_2^k$ holds for all $1 \leq k < j$, where $j \geq 3$. Then, by (\ref{lemA.1}), the induction hypothesis, and (\ref{A.3}), we have $$B_0^j = B_0^{j-2}\cdot\{0^3\} \subseteq B_2^{j-2}\cdot\{0^3\} = B_2^j.$$ This completes the proof that $B_0^j \subseteq B_2^j$ holds for all $j \in \zp$.

By (4.4) and (4.6), we have $B_1^1 = B_2^1$ and $B_1^2 \subseteq B^2_2$. Assume that $B^k_1 \subseteq B^k_2$ holds for all $1 \leq k < j$, where $k \geq 3$. Then, by (\ref{lemA.2}), the induction hypothesis, and (\ref{lemA.3}), we have $$B^j_1 = B^{j-2}_1 \cdot \{0^3\} \subseteq B_2^{j-2} \cdot \{0^3\} = B^j_2.$$ This completes the proof that $B_1^j \subseteq B_2^j$ holds for all $j \in \zp$.
\end{proof} \qed

For each set $B \subseteq \{0\}^*$, let $$B/0 = \{y \in \{0\}^* \mid y0 \in B\}.$$

\begin{lemma}\label{lemA.3}
\begin{flalign*}
\quad &1. \text{ For all even } j \in \zp,\ B_1^j \cdot \{0\} \subseteq B_2^{j}.& \tag{A.4} \label{A.4}\\
\quad &2.\text{ For all } 2 \leq j \in \zp,\ B_1^j /0 \subseteq B_2^{j}.& \tag{A.5} \label{A.5}\\
\quad &2.\text{ For all } j \in \zp,\ B_1^j\cdot \{ 0 \} \subseteq B_2^{j}.& \tag{A.6} \label{A.6}\\
\end{flalign*}
\end{lemma}

\begin{proof}
All three parts of this proof use induction on $j$.

1. By (4.4) and (4.6), $B_1^2 \cdot \{0\} \subseteq B_2^2$. Assume that  $B_1^j \cdot \{0\} \subseteq B_2^j,$ where $j \in \zp$ is even. Then (\ref{lemA.2}), the induction hypothesis, and (\ref{lemA.3}) tell us that

\begin{equation*}
\begin{split}
B_1^{j+2} \cdot \{0\} & =  B_1^j \cdot \{0^4\}\\
 & \subseteq B_2^j \cdot \{0^3\}\\
 &= B_2^{j+2}.
\end{split}
\end{equation*}

2. By (4.4) and (4.6), $B_1^2/ 0 \subseteq B_2^2$. By (4.5), (4.6), and (4.7),

\begin{equation*}
\begin{split}
B_1^3 / 0 & = B_1^2\cdot\{\lambda,0\} \\
 & \subseteq B_2^2 \cdot \{0\}\\
 & \subseteq B_2^3.
\end{split}
\end{equation*}
Assume that $B_1^k/0 \subseteq B_2^k$ holds for all $2 \leq k < j$, where $j \geq 4$. Then (\ref{lemA.2}), the induction hypothesis, and (\ref{lemA.3}) tell us that

\begin{equation*}
\begin{split}
B_1^j/0 & = (B_1^{j-2}\cdot\{0^3\})/0\\
 & = (B_1^{j-2}/0)\cdot\{0\}\\
 &\subseteq B_2^{j-2}\cdot\{0^3\}\\
 &=B_2^j.
\end{split}
\end{equation*}

3. By (4.2) and (4.6), $B_0^1 \cdot \{0\} = B_2^1.$ By (4.3), (4.2), and (4.6), $B_0^2 \cdot \{0\} \subseteq B_2^2$. Assume that $B_0^k \cdot \{0\} \subseteq B_2^k$ holds for all $1 \leq k < j$, where $j \geq 3$. Then (\ref{lemA.1}), the induction hypothesis, and (\ref{lemA.3}) tell us that $$B_0^j\cdot\{0\} = B_0^{j-2}\cdot\{0^4\} \subseteq B_2^{j-2}\cdot\{0^3\} = B_2^j.$$ \end{proof}\qed

With the above results on Construction \ref{constr4.4} in hand, we return to the set $A$ of Construction $\ref{constr4.1}$.

Since $A$ is an infinite subset of $\{0\}^*$, every string in $\{0\}^*$ has infinitely many $A$-extensions. This implies that $|\aj_{0^n}| = 1$ for all $n \in \N$ and $j \in \zp$. For each such $n$ and $j$, then, let $a_n^{(j)}$ be the (unique) element of $\aj_{0^n}$, i.e., the $j^{\text{th}}$ $A$-extension of $0^n$. Note that, for all $j \in \zp$,

\begin{equation}
    A^{(j)} = \{a_n^{(j)} \mid n \in \N\}. \tag{A.7}\label{A.7}
\end{equation}

Recall the set $I'$ of Construction $\ref{constr4.1}$.
\begin{notation} For each $n \in \N$ with $n \geq 3$, let $$\Delta(n) = \{k \in I' \mid n-3 \leq k < n\}$$ and $$\delta(n) = |\Delta(n)|,$$ noting that $\delta(n) \in \{1,2,3\}$ in any case.
\end{notation}

The following holds by routine inspection of Construction \ref{constr4.1}.

\begin{observation}\label{obsA.4}
Let $3 \leq n \in \N$ and $j \in \zp$.

\begin{flalign}
\quad&1.\text{If } \delta(n) = 1, \text{then } a_n^{(j)}0 = a_{n-1}^{(j)}. &\tag{A.8} \label{eqA.8} \\
\quad&2.\text{If } \delta(n) = 2, \text{then } a_n^{(j)}0^3 = a_{n-3}^{(j+2)}. &\tag{A.9} \label{eqA.9}\\
\quad&3.\text{If } \delta(n) = 3, \text{then } a_n^{(j)}0^2 = a_{n-2}^{(j+2)}. &\tag{A.10} \label{eqA.10}
\end{flalign}

\end{observation}

\begin{lemma}\label{lemA.5} For all $j \in \zp$, \begin{equation}a_0^{(j)} \in B_0^j. \tag{A.11}\label{A.11}\end{equation}
\end{lemma}
\begin{proof}
Note that each $a_0^{(j)}$ is simply the $j^{\text{th}}$ element of $A$. We proceed by induction on $j$.

By (4.2), $a_0^{(1)} = \lambda \in B_0^1$. By (4.3), $a_0^{(2)} \in \{0,0^2\} = B_0^2$. Hence (A.11) holds for $j \in \{1,2\}$. Assume that $a_0^{(k)} \in B_0^k$ holds for all $1 \leq k \leq j$, where $j \geq 2$. It suffices to show that \begin{equation}a_0^{(j+1)} \in B_0^{j+1}.\tag{A.12}\label{A.12}\end{equation}

We have two cases.

\textbf{Case 1}: $j+1$ is even. Then $a_0^{(j+1)}$ is $a_0^{(j)}0$ or $a_0^{(j)}0^2$. Either way the induction hypothesis and (4.3) tell us that $$a_0^{(j+1)} \in B_0^j \cdot \{0,0^2\} = B_0^{j+1},$$ i.e., that (A.12) holds.

\textbf{Case 2}: $j+1$ is odd. Then either $$a_0^{(j+1)} = a_0^{(j)}0 = a_0^{(j-1)}0^3$$ or $$a_0^{(j+1)} = a_0^{(j)}0^2 = a_0^{(j-1)}0^3,$$ so the induction hypothesis and (\ref{lemA.1}) tell us that $$a_0^{(j+1)} \in B_0^{j-1} \cdot \{0^3\} = B_0^{j+1},$$ i.e., that (A.12) holds.  \qed
\end{proof}
{}
\begin{lemma}\label{lemA.6}
For all $j \in \zp$, \begin{equation}
    a_1^{(j)} \in B_1^j. \tag{A.13}\label{A.13}
\end{equation}
\end{lemma}
\begin{proof}(by induction on $j$). By (\ref{eq4.4}), $a_1^{(1)}\in \{ \lambda, 0 \}= B_1^{1}$ and $a_1^{(2)}=0^2\in B_1^2$. Hence (\ref{A.13}) holds for $j\in \{ 1,2 \}$. Assume that $a_1^{(k)}\in B_1^k$ holds for $1\leq k\leq j$, where $j\geq 2$. It suffices to show that
\begin{equation}
    a_1^{(j+1)} \in B_1^{j+1}. \tag{A.14}\label{A.14}
\end{equation}
We have two cases.

\textbf{Case 1:} $j+1$ is even. Then, by Construction~\ref{constr4.1}, the induction hypothesis, and (\ref{A.2}),
\[
    a_1^{(j+1)}= a_1^{(j-1)}0^3 \in B_1^{j-1}\cdot \{ 0^3 \} = B_1^{j+1}.
\]

\textbf{Case 2:} $j+1$ is odd. Then, by Construction~\ref{constr4.1}, $a_1^{(j+1)}\in \{ a_1^{j}0, a_1^{(j)}0^2 \}$. It follows by the induction hypothesis and (\ref{eq4.5}) that
\[
     a_1^{(j+1)}\in B_1^{j}\cdot \{ 0, 0^2 \} = B_1^{j+1}.
 \]

 In either case, (A.14) holds. \qed
\end{proof}

\begin{lemma}\label{lemA.7}
For all $j \in \zp$, \begin{equation}\label{eqA.15}
B_1^{(j+1)} \subseteq B_2^j\cdot \{ 0 \}. \tag{A.15}
\end{equation}
\end{lemma}

\begin{proof}
Let $j \in \zp$. We have two cases.\\
\textbf{Case 1}: $j$ is even. Then Lemma \ref{A.2} and (\ref{eq4.7}) tell us that $$B_1^{j+1} \subseteq B_2^{j+1} = B_2^{j} \cdot \{0\},$$ so (\ref{eqA.15}) holds.\\
\textbf{Case 2}: $j$ is odd. Then (\ref{A.2}), (\ref{A.4}), and (\ref{eq4.7}) tell us that

\begin{align*}
 B_1^{j+1} &= B_1^{j-1} \cdot \{0^3\}\\
 &\subseteq B_2^{j-1}\cdot\{0^2\}\\
 &= B_2^j\cdot\{0\}, \notag
\end{align*} so (\ref{eqA.15}) again holds.
\end{proof} \qed

\begin{lemma}\label{lemA.8} For all $j \in \zp$, \begin{equation}\label{eqA.16}
a_2^{(j)} \in B_2^j
\tag{A.16}
\end{equation}
\end{lemma}

\begin{proof}(by induction on $j$). By (\ref{eq4.6}), $a_2^{(1)} \in \{\lambda, 0\} = B_2^1$ and $a_2^{(2)} \in \{0,0^2,0^3\} = B_2^2$. Hence (\ref{eqA.16}) holds for $j \in \{1,2\}$.

Assume that $a_2^{(k)} \in B_2^k$ holds fo $1 \leq k < j$, where $j \geq 2$. It suffices to show that \begin{equation}\label{eqA.17} a_2^{(j+1)} \in B_2^{j+1}. \tag{A.17}\end{equation}

By Constructions \ref{constr4.1} and \ref{eq4.4}, we have $a_2^{(j+1)}0 \in B_1^{j+1} \cup B_1^{j+2}$, whence \begin{equation}\label{eqA.18}
a^{(j+1)} \in (B_1^{j+1}/0) \cup (B_1^{j+2}/0).
\tag{A.18}
\end{equation}
By Lemma \ref{A.3}, \begin{equation}\label{eqA.19} B_1^{j+1}/0 \subseteq B_2^{j+1}. \tag{A.19}\end{equation}
By Lemma \ref{A.7}\begin{equation}\label{eqA.20} B_1^{j+2}/0 \subseteq B_2^{j+1}. \tag{A.20}\end{equation}
By (\ref{eqA.18}),(\ref{eqA.19}), and (\ref{eqA.20}), we have (\ref{eqA.17}).
\end{proof} \qed

\textbf{Proof of Lemma 4.5}
Our main objective is to prove that, for all $j \in \zp$ and $n \in \N$, \begin{equation}\label{eqA.21}
    a_n^{(j)} \in B_{n \mybmod 3}^j. \tag{A.21}
\end{equation} This suffices, because by (\ref{A.7}), Lemma \ref{A.2}, and (\ref{eq4.8}), it implies that \begin{align*}
 \aj &= \{a_n^{(j)} \mid n \in \N\}\\
 &\subseteq \mcup_{n=0}^\infty B_{n \mybmod 3}^j\\
 &= B_0^j \cup B_1^j \cup B_2^j\\
 &= B_2^j\\
 &= B^{(j)}, \notag
\end{align*}
affirming the lemma.

Let $j \in \zp$. We prove (\ref{eqA.21}) by induction on $n$. Lemmas \ref{A.5}, \ref{A.6}, and \ref{eqA.8} tell us that (\ref{eqA.21}) holds for $n\in \{0,1,2\}$.

Assume that $a_m^{(j)} \in B^j_{m \mybmod 3}$ holds for all $0 \leq m < n$, whence $n \geq 3$. It suffices to show that \begin{equation}\label{eqA.22}
a_n^{(j)} \in B_{n \mybmod 3}^j.
\tag{A.22}\end{equation}
We have three cases.

\textbf{Case 1}: $\delta(n) = 1$. Then $n \equiv 2\pmod 3$, so (\ref{eqA.8}), the induction hypothesis, and (\ref{A.5}) tell us that

\begin{align*}
 a^{(j)}_n0 &= a_{n-1}^{(j)} \\
 &\in B_{(n-1)\mybmod 3}^j\\
 &= B_1^j\\
 &= (B_1^j/0)\cdot\{0\}\\
 &\subseteq B_2^{j}\cdot\{0\}, \end{align*}
 whence $a^{(j)}_n \in B^j_2 = B^{j}_{n \mybmod 3}.$

 \textbf{Case 2:} $\delta(n) = 2$. Then (\ref{eqA.9}), the induction hypothesis, and (\ref{A.3}) tell us that

 \begin{align*}
 a^{(j)}_n0^3 &= a_{n-3}^{(j+2)} \\
 &\in B_{(n-3) \mybmod 3}^{j+2}\\
 &= B_{n \mybmod 3}^{j+2}\\
 &= B_{n \mybmod 3}^j\cdot\{0^3\}, \end{align*} whence $a_n^{(j)} \in B_{n \mybmod 3}^{j}$.

 \textbf{Case 3}: $\delta(n) = 3$. Then (\ref{eqA.10}), the induction hypothesis, (\ref{A.2}), and (\ref{A.6}) tell us that

 \begin{align*}
 a^{(j)}_n0^2 &= a_{n-2}^{(j+2)} \\
 &\in B_{(n-2) \mybmod 3}^{j+2}\\
 &= B_0^{j+2}\\
 &= (B_0^j)\cdot\{0\}\cdot\{0^2\}\\
 &\subseteq B_2^{j}\cdot\{0^2\}, \end{align*}
 whence $a_n^{(j)} \in B_2^j = B^j_{n \mybmod 3}$.

 In any case, then, (\ref{eqA.22}) holds. \qed

\end{document}